\def\final{1}
\def\lip{1}

\ifnum\lip=1

\documentclass[a4paper,USenglish]{lipics}
 
\title{Making Existential-Unforgeable Signatures Strongly Unforgeable in the Quantum Random-Oracle Model}
\titlerunning{EU to SU in QRO} 

\author[1]{Edward Eaton}
\author[1,2]{Fang Song}
\affil[1]{Department of Combinatorics \& Optimization, University of Waterloo\\
  \texttt{\{eeaton,fang.song\}@uwaterloo.ca}}
\affil[2]{Institute for Quantum Computing, University of Waterloo\\
 }

\authorrunning{E. Eaton and F. Song} 

\Copyright{Edward Eaton and Fang Song}

\subjclass{E.3 Public key cryptosystems}
\keywords{digital signatures, strongly unforgeable, quantum random-oracle, lattices}

\volumeinfo{Salman Beigi and Robert Koenig}
{2}
{10th Conference on the Theory of Quantum Computation, Communication
and Cryptography (TQC 2015)}
{44}
{1}
{1}
\EventShortName{TQC'15}%
\DOI{10.4230/LIPIcs.TQC.2015.p}
\serieslogo{}

\usepackage{enumerate}

\newenvironment{ncprot}[2]{
\begin{figure}[!ht]
\begin{center}
   \begin{tabular}{|ll|}
   \hline
    \hspace{.1ex} \begin{minipage}{.9\linewidth}\vspace{0.5ex}
       {\begin{center} {\bf #1} {#2} \end{center}}\vspace{-3ex}
       }{%
       \vspace{-2ex}
       \smallskip
     \end{minipage} & \\
     \hline
   \end{tabular}
   \end{center}
   \vspace{-3ex}
\end{figure}
}

\else

\documentclass[11pt]{article}

\usepackage{amsmath,amsthm,amsfonts,amssymb,graphicx,dsfont,bm}

\usepackage[noblocks]{authblk}
\usepackage{color}
\usepackage[normalem]{ulem}
\usepackage{verbatim,enumerate}
\usepackage{hyperref}

\newtheorem{theorem}{Theorem}

\newtheorem{claim}{Claim}
\newtheorem{lemma}{Lemma}
\newtheorem{definition}{Definition}

\newenvironment{ncprot}[2]{
\begin{figure}[!ht]
\begin{center}
   \begin{tabular}{|ll|}
   \hline
     \hspace{.1ex}\begin{minipage}{.97\linewidth}\vspace{0.5ex}
       {\begin{center} {\bf #1} {#2} \end{center}}\vspace{-2ex}
       }{%
       \vspace{-2ex}
       \smallskip
     \end{minipage}& \\
     \hline
   \end{tabular}
   \end{center}
   \vspace{-2ex}
\end{figure}
}
\fi 

\newcommand\mypar[1]{\noindent \textbf{#1}.} 

\ifnum\final=1
\newcommand\fsnote[1]{}
\newcommand\tenote[1]{}
\else
\newcommand\fsnote[1]{[{\color{blue}Fang's note}: #1]}
\newcommand\tenote[1]{[{\color{blue}Ted's note}: #1]} 
\fi
\newcommand\fstd[1]{{\color{red} Next}: #1} 

\newcommand\ket[1]{|#1\rangle}

\def\calO{\mathcal{O}}
\def\calR{\mathcal{R}}
\def\calH{\mathcal{H}}
\def\ws{\text{\sf WS}}
\def\wso{\text{\sf WS}^\calO}
\def\calA{\mathcal{A}}
\def\calB{\mathcal{B}}
\def\calC{\mathcal{C}}

\def\sam{\text{\sf Samp}}

\def\veps{\varepsilon}
\def\adv{\text{\sc ADV}}
\def\negl{\text{\sf negl}}
\def\hyd{\text{\sf Hyd}}

\def\eucma{\text{\sf eu-acma}}
\def\sucma{\text{\sf su-acma}}
\def\too{\text{\sf TOO}}
\def\evt{\text{\sf EVT}}

\begin{document}

\ifnum\final=0

\textbf{\large{To-Do (Future)}}
\begin{itemize}
  \item @Fang inquire BDF13: SU. 
  \item @all non-standard Chameleon property, i.e., sampling $r\in_R \calR$, is it also needed in classical proof? any other hash functions having this property?
  \item @Ted [TOO] journal: chameleon hash based? 
  \item @all quantum-safe SU summary. Plain model: [Ruc10], [BZ13] Constr.3.10 Strong EURMA. QRO: [BZ13](const.3.12 UURMA[can be obtained from PSF+TDP?]; Constr. 3.18 Det. GPV sig needs no rand?), [BDF13], [Unr15].   
\end{itemize}

\newpage
\fi

\ifnum\lip=0
\title{Making Existentially Unforgeable Signatures Strongly Unforgeable in the Quantum Random-Oracle Model}

\author[1]{Edward Eaton, {eeaton@uwaterloo.ca}}
\author[1,2]{\authorcr Fang Song, fang.song@uwaterloo.ca}
\affil[1]{Department of Combinatorics \& Optimization, University of Waterloo}
\affil[2]{Institute for Quantum Computing, University of Waterloo}
\date{}
\fi

\maketitle


\begin{abstract}
Strongly unforgeable signature schemes provide a more stringent security guarantee than the standard existential unforgeability. It requires that not only forging a signature on a new message is hard, it is infeasible as well to produce a new signature on a message for which the adversary has seen valid signatures before. Strongly unforgeable signatures are useful both in practice and as a building block in many cryptographic constructions. 

This work investigates a generic transformation that compiles any existential-unforgeable scheme into a strongly unforgeable one, which was proposed by Teranishi et al.~\cite{TOO06} and was proven in the classical random-oracle model. Our main contribution is showing that the transformation also works against \emph{quantum} adversaries in the \emph{quantum} random-oracle model. We develop proof techniques such as adaptively programming a quantum random-oracle in a new setting, which could be of independent interest. Applying the transformation to an existential-unforgeable signature scheme due to Cash et al.~\cite{CHKP12}, which can be shown to be quantum-secure assuming certain lattice problems are hard for quantum computers, we get an efficient quantum-secure strongly unforgeable signature scheme in the quantum random-oracle model. 
\end{abstract}

\section{Introduction}
\label{sec:intro}

Digital signature is a fundamental primitive in modern cryptography and has numerous applications. In a signature scheme, a signer uses his/her secret key to generate a signature on a message. Anyone who knows the corresponding public key can verify the integrity of the message and that it comes from the genuine signer. A standard security notion for digital signatures is called \emph{existential-unforgeable} under \emph{adaptive chosen-message-attacks} ($\eucma$ in short). Basically it means that an adversary, without knowing the secret key of a user, cannot forge a valid signature on a \emph{new} message. This should hold even if the adversary has seen a few signatures generated by the honest user on messages adaptively chosen by the adversary. Another important security notion, stronger than $\eucma$, is called \emph{strongly} existential-unforgeable ($\sucma$). Here, in addition to $\eucma$, it should be infeasible to forge a \emph{new} signature on a previously signed message. Aside from applications in some practical scenarios~\cite{Ruc10}, $\sucma$ signatures turn out to be a very powerful tool in other cryptographic constructions. For instance they are used in transforming encryption schemes that are secure under chosen-plain-text attacks into secure schemes under \emph{chosen-ciphertext-attacks}~\cite{DDN00,BCHK06}; and in constructing identity-based blind signatures~\cite{GHK06} and group signature schemes~\cite{ACJT00,BBS04}. 

Strongly unforgeable signature schemes can be obtained from existential-unforgeable ones via generic transformations~\cite{SPW06,HWZ07,TOO06}. The transformation in~\cite{TOO06} (referred to as $\too$ hereafter) is particularly interesting because it only needs a mild computational assumption and the overhead it causes to the efficiency is small. This work studies this transformation in the quantum setting, where adversaries have the power of processing quantum information. We want to ask: 
\emph{does $\too$ transformation still hold in the presence of quantum adversaries, and furthermore can we obtain quantum-secure $\sucma$ signatures systematically?} 

There is no quick answer to this question. In general a classically secure cryptographic construction can completely fall apart against quantum adversaries for at least two reasons. First of all, quantum computers can solve some problems efficiently which are otherwise believed hard classically. This breaks the computational assumption in many constructions. For example, many existing $\eucma$ signature schemes, the starting point of the transformation, are based on factoring or discrete logarithm. The $\too$ transformation itself also uses the discrete logarithm problem. They are immediately broken by Shor's quantum algorithms~\cite{Shor97}. Naturally we may want to switch to \emph{quantum-safe} assumptions. For example, we assume certain lattice problems are hard even against quantum algorithms and then construct crypto-systems based on them~\cite{MR09,BBD09}. However, this does not fix everything immediately due to another reason, which is more subtle. Security of a construction is established by a security reduction, which is a proof by contradiction showing that if a scheme is not secure, then one can break a computational assumption. Unfortunately, as pointed out by a line of works (e.g.,~\cite{Wat09,HSS11,Unr12pok,Son14}), classical security reductions may not hold in the presence of quantum adversaries due to technical difficulties such as \emph{quantum rewinding}.  

There is an additional complication, which turns out to be the main difficulty towards making the $\too$ transformation go through in the quantum setting. Classically, $\too$ is proven in the random-oracle model (RO), where a hash function is treated as a truly random function and all users evaluate the hash function by querying the random function
. However once an adversary becomes quantum, we should naturally allow the queries to be in quantum superposition. This is formalized as the quantum random-oracle model (QRO)~\cite{BDFLSZ11}. The bad news is that many classical tricks in RO become difficult to apply in QRO, if not entirely impossible. For starters, classically it is trivial to answer random-oracle queries on-the-fly by generating fresh random value for new queries while maintaining a table to keep consistency. It is not obvious that some similar trick can handle quantum superposition queries. There have been a host of works in recent years developing proof techniques in QRO~\cite{Zha12a,Unr14_rev,Unr14_pos}, but many classical techniques are still missing their counterparts in QRO. 

\mypar{Our Contributions} Our main result is showing that the $\too$ transformation still works against quantum adversaries in the quantum random-oracle model under reasonable computational assumptions. Specifically, we first make a simple observation that classically the $\too$ transformation actually holds using any (generic) chameleon hash function, rather than the specific instantiation by the discrete log problem.  As our central contribution, we prove that once the chameleon hash function and the $\eucma$ signature scheme are both quantum-safe, then $\too$ transformation will produce a quantum-safe $\sucma$ signature scheme in the quantum random-oracle model. In our proof, we develop a technique that allows for adaptively programming a quantum random-oracle in a new setting. We hope this technical can find applications and extensions elsewhere. 


Once we have the transformation ready, we demonstrate instantiations of the building blocks to obtain concrete quantum-safe $\sucma$ schemes. Using tools from~\cite{Son14}, it is easy to verify that the bonsai-tree signature scheme by Cash et al.~\cite{CHKP12} is $\eucma$ against quantum adversaries assuming some lattice problem is quantum-safe\footnote{Actually, we observe a tighter security reduction so that a slightly weaker  assumption on the lattice problem is sufficient.}. In~\cite{CHKP12}, a chameleon hash function was also proposed based on the same computational assumptions, which is easy to check that it is quantum-safe as well. Putting these pieces together, we can get a quantum-safe $\sucma$ scheme. 


\mypar{Overview of Our Proof Techniques in QRO} As we mentioned earlier, many proof techniques in classical RO do not immediately go through in the QRO model. Roughly speaking, the classical proof for the $\too$ transformation relies on two features in the classical RO model: the history of queries that an adversary makes to the RO can the recorded, and at various steps one can assign a fresh random value on an input, since the response at an input needs not to be determined before being queried. Both become difficult in the quantum setting. Copying quantum superposition queries which are unknown quantum states is generally impossible, and apparently a single quantum query of the form $\sum|x,y\rangle \mapsto \sum |x,\calO(x)\oplus y\rangle$ would ``see'' the function values at all inputs. It is hence unclear how to change $\calO(x)$ later without being caught. 

The first issue turns out to be  non-essential. The purpose of keeping the RO queries is to make sure some special input $x^*$ has not been queried by the adversary. Otherwise $x^*$ can be used to break some assumption. In the quantum setting, we can just pick one of the queries at random and measure it. If the overall amplitude that adversary intends to query at $x^*$ is high, the probability we recover $x^*$ is only reduced by essentially a poly-factor (the number of the adversary's RO queries).

We then come up with a technique for adaptively programming a QRO in a new setting. Namely we want to change the function value at various inputs that the adversary has partial control (e.g., the prefix of these inputs are chosen by the adversary). Intuitively this is possible when these inputs still have sufficient uncertainty to the adversary. There exist techniques previously when these input strings are \emph{information-theoretically} undetermined, possessing a high min-entropy for example~\cite{Unr14_pos,Unr15}.  In contrast, in our case these inputs are \emph{computationally} difficult to decide by the adversary. Namely, these inputs remain uncertain to the adversary unless some computational assumption is broken. We show that this is already sufficient freedom for programming the answers on these inputs. Being a little more specific, 
we show that the computational assumption implies \emph{indistinguishability} of two functions which a distinguisher can have quantum access to: one is the all-zero function, and the other marks a set of strings that could be used to break the computational assumption. This may be interpreted as a computational analogue of the Grover search lower bound in quantum query complexity. This enables us to program a quantum random-oracle adaptively. Basically, the random-oracle embeds one of the preceding functions, and programming the random-oracle roughly amounts to switching between the two functions. Since the two functions are indistinguishable, any efficient quantum algorithm querying the random-oracle cannot notice whether we have re-programmed the quantum random-oracle. From a technical point of view, these claims may not sound very surprising. Nonetheless, we view them as an interesting conceptual shift, which is similar in spirit to~\cite{CDMS04} where the authors showed that \emph{computational} constraints can force measurement on a quantum state and cause collapse to particular subspaces. Our techniques also complements existing ones that are of information-theoretical flavor. 

\mypar{Related Works} Boneh and Zhandry~\cite{BZ13b} considered a stronger type of quantum attacks on signature schemes where an adversary can query a signing oracle in superposition. They proposed general transformations which  amplify schemes that are secure against ordinary quantum adversaries (i.e., those who only issue classical signing query as we consider in this work), to achieve security under attacks with superposition signing queries. In contrast, the transformation in our work only considers ordinary quantum adversaries, but tries to amplify in terms of the type of forgeries that an adversary can produce. Lyubashevsky~\cite{Lyu09,Lyu12} applied the Fiat-Shamir paradigm to construct lattice-based $\sucma$ signatures in the random-oracle model from identification schemes. However whether these schemes are quantum-secure is unclear, because proving quantum security of the identification schemes faces the difficulty of quantum rewinding. More importantly, there is negative evidence that Fiat-Shamir paradigm may not hold in general in the QRO model~\cite{DFG13,ARU14}. Dagdelen et al.~\cite{DFG13} showed that a variant of Fiat-Shamir works in the QRO  model, but only for a very special form of identification schemes. In a recent work by Unruh~\cite{Unr15}, a general transformation is proposed, which can produce (quantum-safe) strongly-unforgeable signatures in the QRO model from general $\Sigma$-protocols. However the overhead is much larger than the Fiat-Shamir transformation, and the resulting signature schemes are less efficient than what can be obtained from our work.  
\fsnote{Ask the authors! Dagdelen et al.~\cite{DFG13} proved that a specific variant of~\cite{Lyu12} can be made existential-unforgeable via the Fiat-Shamir transformation in the QRO model. However we do not know if strong unforgeability can be achieved.}
We remark that there is a generic Merkle-tree approach that produces $\sucma$ schemes out of $\sucma$ one-time signature schemes, which should still hold against quantum adversaries. Therefore in principle, lattice-based one-time signatures, as in~\cite{LM08}, would suffice for full-fledged quantum-safe $\sucma$ schemes. However the resulting scheme is usually far less efficient and costly to manage (because it is typically stateful). 

\section{Preliminary}
\label{sec:prelim}

We review necessary definitions and cryptographic tools in this section. 

\begin{definition}[Signature Scheme]
A \textbf{signature scheme} is composed of a triplet of probabilistic polynomial-time algorithms $(G,S,V)$, satisfying the following:

\begin{itemize}
  \item $G$ is the key generation algorithm. On running, it produces a pair, $(pk, sk)$. $pk$ is the public key, or verification key, while $sk$ is the secret key, or signing key.
  \item $S$ is the signing algorithm. Upon input of a message $M$ from a message space $\mathcal{M}$, as well as a secret key $sk$, it produces a signature $\sigma$ on that message.
  \item $V$ is the verification algorithm. It takes in a message $M$, a signature $\sigma$, and a public key $pk$, and will output either `accept' or `reject'.
\end{itemize}

Signature schemes must satisfy the \textbf{correctness requirement}, which is that for any $(pk, sk)$ generated by $G$, and any $M \in \mathcal{M}$, if $\sigma \leftarrow S(M, sk)$ then $V(M, \sigma, pk) =$ `accept'.
\end{definition}

A standard security notion for signature schemes is \textbf{existential unforgeability under adaptive chosen message attack} ($\eucma$).

\begin{definition}[Existential Unforgeability under Adaptive Chosen Message Attack]

Consider the following game between a challenger $\mathcal{C}$ and a forger $\mathcal{A}$:
\begin{itemize}
\item $\mathcal{C}$ runs $G$, and send the resulting $pk$ to $\mathcal{A}$.
\item $\mathcal{A}$ sends up to $q$ messages $M_1, M_2, ... , M_q$ to $\mathcal{C}$, one at a time. For each message $\mathcal{C}$ receives, she sends back $\sigma_i = S( M_i , sk )$ to $\mathcal{A}$.
\item $\mathcal{A}$ finally outputs a pair $(M^* , \sigma^* )$ to $\mathcal{C}$. We call this a valid forgery if $M^* \neq M_i \forall i \in \{1, ... , q \}$ and $V(M^*, \sigma^*, pk) =$ `accept'.
\end{itemize}

If, for polynomially bounded $q$, it is computationally infeasible for $\mathcal{A}$ to come up with a valid forgery, the scheme is said to be existentially unforgeable under adaptive chosen message attack.
\end{definition}

\begin{definition}[Strong Unforgeability under Adaptive Chosen Message Attack]

\textbf{Strong unforgeability under Adaptive Chosen Message attack}, or \textbf{su-acma}, is defined in the same way as eu-acma, except that the pair $(M^*, \sigma^* )$ that $\mathcal{A}$ eventually submits must only require that $(M^*, \sigma^*) \neq (M_i , \sigma_i)$ for all $i$, instead of the requirement that $M^* \neq M_i$. This change means that the forgery $\mathcal{A}$ submits may either be a new message, or may be a message that $\mathcal{C}$ has already signed, but with a new signature. \end{definition}
Note that by allowing $\mathcal{A}$ to submit more kinds of forgeries, if it is still computationally infeasible for $\mathcal{A}$ to succeed, then we know that this type of forgery also cannot be created, making the scheme in a sense stronger. 

\mypar{Chameleon hash functions} Chameleon hash functions were introduced by Krawczyk and Rabin~\cite{KR00}. We need a slight generalization proposed in~\cite{CHKP12}. A family $\calH$ of chameleon hash function is a collection of functions $h$ that takes in a message $m$ from a message space $\mathcal{M}$ and some randomness $r$ from a randomness space $\mathcal{R}$, and outputs to a range $\mathcal{Y}$, ie, $h: \mathcal{M} \times \mathcal{R} \rightarrow \mathcal{Y}$. The randomness space is associated with some efficiently sampleable distribution. There are three properties we need for a family of chameleon hash functions:
\begin{itemize}
	\item (Chameleon property) We require an algorithm $HG$ that samples a hash function $h\in \calH$ together with trapdoor information $td$ satisfying that for any $m \in \mathcal{M}$ and $y \in \mathcal{Y}$, it is possible to efficiently sample $r \gets h^{-1}_{td}(m,y)$ under the distribution associated with $\mathcal{R}$ such that $h(m, r) = y$. \fsnote{Caution: this may be crucial}
	\item (Uniformity) For $h \gets \calH$ and $r\gets \mathcal{R}$, $(h, h(m,r))$ is uniform over $(\mathcal{H}, \mathcal{Y})$ up to negligible statistical distance. 
	\item (Collision resistance) For a hash function $h \gets \calH$, it is computationally infeasible for an adversary to find $(m, r), (m', r')$, with $(m, r) \neq (m', r')$ such that $h(m, r) = h(m', r')$.

\end{itemize} 






\mypar{Quantum Random-Oracle Model} The random oracle model is a technique used in cryptographic proofs. In it, Hash functions are replaced with random oracles. An adversary is  given access to query this random oracle by providing an input, $x$, and is returned the response, $\mathcal{O}(x)$. These random oracles exist to replace hash functions in our proof. When we examine the proof in the context of quantum computers, Boneh et al.~\cite{BDFLSZ11} have pointed out that since superposition queries to hash functions are possible, to truly capture this in a model allowing quantum computers, we must allow superposition queries to the random oracle. So we will allow superpositions of queries to our random oracle, $\sum a_x | x, y \rangle$, which will be responded to with a superposition of answers, $\sum a_x | x, y \oplus \calO(x) \rangle$.

A cryptographic scheme is said to be \emph{quantum-safe} (or quantum-secure) if the security conditions still hold once the adversaries become efficient quantum computers. We do not go into more precise definitions. See for example~\cite{HSS11} for details.

\section{Getting SU from EU in QRO}
\label{sec:intro}

In this section we prove our main theorem. 

\begin{theorem}
There exists a generic conversion that takes an quantum-safe $\eucma$ signature scheme $\Sigma = ( G, S, V)$ and a family of quantum-safe collision-resistant chameleon hash functions $\calH$ and produces a quantum-safe $\sucma$ signature scheme $\Sigma ' = (G', S', V' )$ in the quantum random-oracle model. 
\label{thm:eu2su}
\end{theorem}


\subsection{The Transformation}
\label{subsec:trans}

We first recall the $\too$ transformation~\cite{TOO06} with a minor change. We use a generic chameleon hash function instead of an instantiation from the discrete log problem.

\begin{itemize}
\item $G'$. On input a security parameter $1^n$, do the following: 
\begin{itemize}
\item Run $G$, obtaining $(pk, sk)$. 
\item Run $HG$ obtaining a chameleon hash function $h$ with trapdoor $td$. 
\item Set $pk' = (pk, h)$ and $sk' = (sk, td)$.
\end{itemize}
\item $S'$. On input of message $M$, do the following:
\begin{itemize}
\item 
Sample a random $C$ from the range of $h$. 
\item Sign $C$ using the signing algorithm $S$, obtaining $\sigma = S(C, sk)$
\item Compute $m = \mathcal{O} (M \|  \sigma)$, where $\mathcal{O}$ is a hash function (to be replaced with a random oracle in the proof).
\item Using the trapdoor information $td$, find an $r$ such that $h(m, r) = C$.
\item Output $\sigma' = ( \sigma, r)$.
\end{itemize}
\item $V'$. On input of a message $M$ and a signature $\sigma' = (\sigma, r)$, do the following:
\begin{itemize}
\item Compute $m = \mathcal{O}(M \| \sigma)$ and $C = h(m, r)$.
\item Output 'Accept' if and only if $V(C, \sigma, pk) =$ 'Accept' (otherwise, output 'Reject').
\end{itemize}
\end{itemize}

The correctness of the algorithm can be seen easily. If $\sigma'$ was a signature generated on $M$ using $S'$, then $C$ will be the same $C$ generated during the running of $S'$, and is precisely what $\sigma$ is a signature for.

\subsection{Main Technical Lemma: Adaptively Programming a Quantum RO}
\label{ssec:apqro}

To prove the main theorem, we demonstrate a new scenario where we can adaptively program a quantum random-oracle. This extends existing works (e.g~\cite{Unr14_pos,Unr14_rev,Unr15}) from information-theoretical setting to a computational setting, and we believe it is potentially useful elsewhere. We will formalize a probabilistic game which we call \emph{witness-search}. It potentially captures the essence of numerous security definitions for cryptographic schemes (e.g. signatures). Then we show that the (computational) hardness of witness-search allows for adaptively programming a quantum random-oracle. 

Let $\sam$ be an instance-sampling algorithm. On input $1^n$, $\sam$ generates public information $pk$, description of a predicate $P$, and a witness $w$ satisfying $P(pk,w) = 1$. 
Define a {witness-search} game $\ws$ as below.

\begin{ncprot}{Witness-Search Game}{$\ws$}
\begin{enumerate}
	\item Challenger $\calC$ generates $(pk,w,P)\gets\sam(1^n)$. Ignore $w$. Let $W_{pk}: = \{w: P(pk,w) =1\}$ be the collection of valid witnesses. 
	\item $\calA$ receives $pk$ 
	and produces a string $\hat w$ as output. 
	\item We say $\calA$ wins the game if $\hat w\in W_{pk}$. \\
\end{enumerate}
\end{ncprot}
 
We say $\ws(\sam)$ is hard, if for any poly-time $\calA$, $\Pr[\calA \text{ wins}] \leq \negl(n)$. For instance,  $\sam$ could be the KeyGen algorithm of a signature scheme. $pk$ consists of the public key and description of the signature scheme. Predicate $P$ is the verification algorithm and a witness consists of a valid message-signature pair. Security of the signature scheme implies hardness of $\ws(\sam)$. 

\begin{lemma}[Hardness of Witness-Search to Programming QRO]
Let two experiments $E$ and $E'$ be as below. If $\ws$ is hard, then $\adv:= \left| \Pr_{E}[b=1] - \Pr_{E'}[b=1] \right| \leq \negl(n)\, .$
\label{lemma:ws2pro}
\end{lemma}

Note that $E'$ differs from $E$ only in that we reprogram the random oracle at some point in $E'$. We defer the proof of this lemma to Appendix~\ref{sec:pfadp}. 

\begin{ncprot}{Experiment}{$E$ }
\begin{enumerate}
	\item  Generate $(pk,w,P)\gets \sam(1^n)$. 
	\item $\calO\gets \mathcal{F}$ is drawn uniformly at random from the collection of all functions $\mathcal{F}$.
	\item $\calA_1$ receives $pk$ as input and makes at most $q_1$ queries to $\calO$. $\calA_1$ produces a classical string $x$. 
	\item Set $z:=\calO(x\|w)$.
	\item $\calA_2$ gets $(x,w,z)$ and may access the final state of $\calA_1$. $\calA_2$ makes at most $q_2$ queries to $\calO$. It outputs $b \in \{0,1\}$ at the end. \\
\end{enumerate}
\end{ncprot}

\begin{ncprot}{Experiment}{$E'$}
\begin{enumerate}
	\item  Generate $(pk,w,P)\gets \sam(1^n)$. 
	\item $\calO\gets \mathcal{F}$ is drawn uniformly at random from the collection of all functions $\mathcal{F}$.
	\item $\calA_1$ makes at most $q_1$ queries to $\calO$. It produces a classical string $x$. 
	\item Pick a random $z\in_R \text{\sf Range}(\calO)$. Reprogram $\calO$ to $\calO'$: $\calO'(y) = \calO(y)$ except that $\calO'(x\|w) = z$. 
	\item $\calA_2$ gets $(x,w,z)$ and may access the final state of $\calA_1$. $\calA_2$ makes at most $q_2$ queries to $\calO'$. It outputs $b\in \{0,1\}$ at the end. \\
\end{enumerate}
\end{ncprot}

To prove Lemma~\ref{lemma:ws2pro}, we need another lemma below to pave the road. Roughly we want to argue that if witness-search is hard, then given an oracle which is either the all-zero function or a function that marks the witness set $W_{pk}$, no efficient algorithms can distinguish them. This may be intuitively interpreted as a computational analogue of Grover search lower bound. Its proof can be found in Appendix~\ref{sec:pfs2ind}.

\begin{lemma}
Let $f$ be the all-zero function, and $f_S$ be the characteristic function of a set $S$. Namely $f_S(x) =1$ iff. $x\in S$. Define two experiments $G$ and $G'$ as below.  If $\ws(\sam)$ is hard, then for any efficient $\calA$ making $q\leq poly(n)$ queries, $\left| \Pr_{G}[b=1] - \Pr_{G'}[b=1] \right| \leq \negl(n)$. 

\begin{ncprot}{Experiment}{$G$}
\begin{enumerate}
	\item Generate $(pk,w,P)\gets \sam(1^n)$. 
	\item $\calA$ is given $pk$ and (quantum) access to $f$. $\calA$ makes at most $q$ queries to $f$ and afterwards $w$ is given to $\calA$. It outputs $b\in\{0,1\}$ and aborts.\\ 
\end{enumerate}
\end{ncprot}
\begin{ncprot}{Experiment}{$G'$}
\begin{enumerate}
	\item Generate $(pk,w,P)\gets \sam(1^n)$. Let $f_{pk}:=f_{W_{pk}}$, where $W_{pk} = \{w:P(w)=1\}$. (i.e., $f_{pk}(x) =1$ iff. $x\in W_{pk}$)
	\item $\calA$ is given $pk$ and (quantum) access to $f_{pk}$. $\calA$ makes at most $q$ queries to $f_{pk}$ and afterwards $w$ is given to $\calA$. It output $b\in\{0,1\}$ and aborts.\\ 
\end{enumerate}
\end{ncprot}

\label{lemma:s2ind}
\end{lemma}

\begin{proof}[Proof of Lemma~\ref{lemma:ws2pro}]

We use a hybrid argument to prove the theorem. Define $E_i, i=1,\ldots, 4$ as follows. 
\begin{itemize}
	\item $E_1 := E$.  (${\calA_1^\calO}/{\calA_2^\calO}$ in short.)
	\item $E_2$: identical to $E_1$ except that in step 3, $\calO$ is replaced by $\bar \calO$ where $\bar \calO(y) = \calO(y)$ but $\bar \calO(y) = 0$ for any $y=\cdot\|w$ where $w\in W_{pk}$.  (${\calA_1^{\bar\calO}}/{\calA_2^\calO}$) \fsnote{$\perp$ or 0?}
	\item $E_3$: identical to $E_2$ except that  after step 3, we use $\calO'$ as defined in $E'$ instead of $\calO$. Observe that $E_3$ can also be obtained from $E'$ by substitute $\bar \calO$ for $\calO$ in step 3. (${\calA_1^{\bar\calO}}/{\calA_2^{\calO'}}$)
	\item $E_4: = E'$. (${\calA_1^{\calO}}/{\calA_2^{\calO'}}$)
\end{itemize}

Define $\adv_i:= \left| \Pr_{E_i}[b=1] - \Pr_{E_{i+1}}[b=1] \right| $. We will show that $\adv_1$ and $\adv_3$ are both negligible using Lemma~\ref{lemma:s2ind}. $\adv_2  = 0$ since in both $E_2$ and $E_3$, the function values for $W_{pk}$ are assigned uniformly at random and independent of anything else. Therefore we conclude that $\adv= \left| \Pr_{E}[b=1] - \Pr_{E'}[b=1] \right| \leq \sum \adv_i = \negl(n)$. 

We are only left to prove that $\adv_1 \leq \negl(n)$, and $\adv_3 \leq \negl(n)$ follows by similar argument. Suppose for contradiction that there exist $(\calA_1,\calA_2)$ such that $\adv_1\geq 1/{p(n)}$ for some polynomial $p(\cdot)$. We show that this will lead to a contradiction to Lemma~\ref{lemma:s2ind} that $|\Pr_G[b= 1] - \Pr_{G'}[b=1]|  \leq \negl(n)$, which in turn contradicts the hardness of witness-search. To see this, we construct an algorithm $D$ from $(\calA_1,\calA_2)$ that runs in $G$ and $G'$ such that $|\Pr_G[b= 1:D] - \Pr_{G'}[b =1:D]| \geq 1/{p(n)}$. Let $F$ be an oracle which ignores the first part of the input and then applies either all-zero function $f$ or $f_{pk}$ (as defined in $G'$) on the second part. Let $g$ be a random function. Define another oracle $H:=g\circ F$ that implements the following transformation: 
\begin{align*}
	\ket{x,y}  \mapsto& \ket{x,y}\otimes\ket{0} \quad \text{append an auxiliary register}\\
	\mapsto & \ket{x,y}\otimes\ket{\overline{F(x)}} \quad \text{compute the negation of $F$ on aux.} \\ 
	\mapsto & \ket{x,y\oplus \overline{F(x)}\cdot g(x)}\otimes\ket{\overline{F(x)}} \quad \text{controlled-$g$} \\
	\mapsto & \ket{x,y\oplus \overline{F(x)}\cdot g(x)} \quad \text{uncompute negation of $F$ and disgard aux.}
\end{align*}
Observe that if $F$ is induced from $f$ then $H$ is identical to a random function $\calO$. Whereas if $F$ comes from $f_{pk}$ then $H$ is identical to $\bar \calO$ as in $E_2$. For an algorithm that queries at most $q$ times to $H$, we can sample $h$ from a family of $2q$-wise independent functions and simulate $H$ efficiently (with access to $F$) without any noticeable difference. 
\begin{ncprot}{Construction of}{$D$}
\begin{enumerate}
	\item $D$ receives $pk$ and an oracle $F$  (one of the two candidates above). 
	\item $D$ simulates oracle $H=g\circ F$ as defined above. $D$ then simulates $\calA_1$, for each of query from $\calA_1$, it is answered by $H$ with (two) oracle calls to $F$. Let $x$ be the output of $\calA_1$. 
	\item $D$ receives $w$ (from external challenger). It then simulates $\calA_2$ on input $(x,w,z:=H(x\| w))$ and oracle queries are answered by $h$. 
	\item $D$ outputs the output of $\calA_2$. \\
\end{enumerate}
\end{ncprot}

It is easy to see that if $F$ is induced from $f$, the view of $\calA_1$ and $\calA_2$ is identical to that of $E_1$. Likewise if $F$ is induced by $f_{pk}$ then it is the same view as in $E_2$. Therefore $|\Pr_G[b= 1:D] - \Pr_{G'}[b=1:D]| = |\Pr_{E_1}[b= 1:(\calA_1,\calA_2)] - \Pr_{E_2}[b=1: (\calA_1,\calA_2)]|   \geq 1/{p(n)}$. This gives a contradiction.  

\end{proof}

\subsection{Proof of Theorem~\ref{thm:eu2su}}
\label{subsec:trans}

\mypar{Brief Review of Classical Proof} Classical proof roughly goes as follows: consider a forger $\mathcal{A}$. If $(M^*, \sigma'^*)$ is the forgery that $\mathcal{A}$ eventually submits, we will let $C^* = h( \mathcal{O}(M^* \| \sigma^* ), r^* )$. Similarly, for a signing query made by the forger $M_i$, we let $C_i = h( \mathcal{O}(M_i \| \sigma_i ), r_i)$.

We then analyze two separate cases. First the instance where $C^* \neq C_i$ for all $i$. In this case we show that this gives a break to the existential unforgeability of the signature scheme $\Sigma$, by way of $(C^*, \sigma^*)$. Next, we examine the case where $C^* = C_i$ for some $i$. In this case we show that $(\mathcal{O}(M^*||\sigma^*), r^*)$ and $(\mathcal{O}(M_i||\sigma_i), r_i)$ provide a break to the collision resistance of the chameleon hash function.

For completeness the full classical proof is included in Appendix~\ref{app:cproof}. It is adapted from~\cite{TOO06} and we use a generic chameleon hash function instead of a concrete instantiation from the discrete logarithm problem. There are also changes which by our opinion make the proof easier to understand.  

\mypar{Proof in the quantum random-oracle model} Let $\mathcal{A}$ be the forger making at most $q$ queries, and let $\epsilon$ be the probability that $\mathcal{A}$ succeeds in her forgery. We construct $\mathcal{B}$ that either breaks existential unforgeablity of $\Sigma$ or can find collisions in $\calH$.

%

\begin{itemize}
\item \textbf{Case 1}: We define this case as occurring when $C^* \neq C_i$ for all $i$.

Firstly, $\mathcal{B}$ will be acting as a quantum random oracle for $\mathcal{C}$. To do this, $\mathcal{B}$ simply chooses a $2q$-wise independent hash function, $\mathcal{O}$, and for any query $\mathcal{A}$ makes, $\Sigma \alpha_{x,z} | x, z \rangle$, $\mathcal{B}$ responds with $\Sigma \alpha_{x,z} | x, \mathcal{O}( x ) \oplus z \rangle$.

\begin{ncprot}{Construction of Existential Forger}{$\calB$}
\begin{enumerate}
	\item $\calB$ receives a public key $pk$ from the challenger $\calC$
	\item $\calB$ simulates a variant of the strongly-unforgeable game with $\calA$:
	\begin{enumerate}[i)]
		\item $\calB$ generates $(h,td)\gets HG(1^n)$. Initiate $\calA$ with $pk'=(pk,h)$
		\item $\calB$ simulates a random-oracle using a $2q$-wise independent hash function. 
		\item On the $i$th signing query $M_i$ from $\calA$, $\calB$ chooses a random $C_i$. It  then signs $C_i$ by submitting it to $\calC$, obtaining $\sigma_i$. It computes $m_i = \calO(M_i||\sigma_i)$, and using the trapdoor information $td$, finds an $r_i$ such that $h(m_i,r_i) = C_i$. It sends $\sigma_i' = (\sigma_i, r_i)$ to $\calA$.
	\end{enumerate}
	\item Let $(M^*,(\sigma^*,r^*))$ be the final forgery produced by $\calA$. Output $(C^*, \sigma^*)$ as the forgery. 
\end{enumerate}
\end{ncprot}

From $\calA$'s point of view, a $2q$-wise independent function is identical to a random function~\cite{Zha12a}. Noting that $C^* \neq C_i$ for all $i$, and the $C_i$'s are precisely what was submitted to $\mathcal{C}$ for signing queries, and finally, seeing as this is a valid forgery, so $V( C^*, \sigma^*) = 'accept'$, we can see that $\mathcal{B}$ submits $(C^*, \sigma*)$ as a valid new forgery, breaking the existential unforgeability of $\Sigma$ and winning his game with $\mathcal{C}$. Thus in this case whenever $\mathcal{A}$ succeeds, so does $\mathcal{B}$, and so the probability $\mathcal{B}$ succeeds given we are in this case is $\epsilon$.

\item \textbf{Case 2}: This case is defined as occurring when $C^* = C_i$ for some $i$. In this case we will show a reduction to break the collision resistance of the chameleon hash function.

\begin{ncprot}{Construction of Collision-Finding Adversary}{$\calB$}
\begin{enumerate}
	\item $\calB$ receives $h$ from the challenger, which is sampled from the Chameleon hash function family. 
	\item $\calB$, playing the role of a challenger, simulates a variant of the strongly-unforgeable game with $\calA$:
	\begin{enumerate}[i)]
		\item $\calB$ generates $(pk,sk)\gets G(1^n)$. Initialize $\calA$ with $pk'=(pk,h)$. For $i = \{1,\ldots,q\}$, $\calB$ generates $m_i$ uniformly at random and $r_i\gets \calR$ (according to the specification of $h$). $\calB$ computes $C_i := h(m_i, r_i)$ and $\sigma_i := S(sk, C_i)$.
		\item $\calB$ simulates a random-oracle in the usual way (i.e. $t$-wise independent hash function). 
		\item On the $i$th signing query $M_i$ from $\calA$, $\calB$ reprograms the random-oracle: $\calO(M_i\|\sigma_i) \gets m_i$ and returns $(\sigma_i,r_i)$ to $\calA$. 
	\end{enumerate}
	\item Let $(M^*,(\sigma^*,r^*))$ be the final forgery produced by $\calA$. We know $C^* = C_i$ for some $i$. Output $(\mathcal{O}(M^*||\sigma^*), r^*), (\mathcal{O}(M_i||\sigma_i), r_i)$ as the collision.  
\end{enumerate}
\end{ncprot}

It is easy to see that $\calB$ finds a valid collision as long as $\calA$ produces a valid forgery, with overwhelming probability. This is because if $C^* = C_i$, then $h(\mathcal{O}(M^*||\sigma^*), r^*) = h(\mathcal{O}(M_i||\sigma_i),r_i)$. We simply need to ensure that this is not a trivial collision. Note that since this must be a new forgery, $(M^*, \sigma^*, r^*) \neq (M_i, \sigma_i, r_i)$. If $r^* \neq r_i$, we are done. Otherwise, we can see that $M^*||\sigma^* \neq M_i||\sigma_i$, and thus since the values for $\calO(M_i||\sigma_i)$ were chosen uniformly at random, $\calO(M^*||\sigma^*) \neq \calO(M_i||\sigma_i)$ with overwhelming probability.

Therefore if we let $\evt$ be the event  that $\calA$ produces a valid forgery,  we only need to show that $\evt$ occurs with probability $\Omega(\veps)$  in the construction of $\calB$. We prove it by a hybrid argument which transforms the standard strongly unforgeable game into the variant as in the construction of $\calB$. We will show that the probablity of $\evt$ is esstially preserved in the hybrid argument. 

Let $\hyd_0$ the standard strongly-unforgeable game with $\calA$. By hypothesis $\Pr[\evt:\hyd_0]\geq \veps$. Consider the first hybrid $\hyd_1$ that makes only one change to $\hyd_0$: when the challenger answers a signing query, instead of querying the random-oracle $\calO$ to obtain $m_i:=\calO(M_i \|\sigma_i)$, it samples a random $m_i$ and programs the random oracle so that $\calO(M_i \| \sigma)= m_i$. Note that in particular the challenger still uses the trapdoor to find $r_i \gets h^{-1}(C_i,m_i)$. By Lemma~\ref{lemma:ws2pro}, we claim that\footnote{More precisely, we need to introduce sub-hybrids and each sub-hybrid makes such a change for just one signing query.} $\Pr[\evt:\hyd_0] - \Pr[\evt:\hyd_1]| \leq \negl(n)$. Specifically we instantiate $\sam$ as follows. $pk$ will consists of a public key for $\Sigma$, hash function $h$, and random messages $C_i$. $P$ will be the verification algorithm of $\Sigma$. $w:=\sigma_i=S(sk,C_i)$ is the signature generated by $\calB$ in 2.i), and $W_{pk}$ consists of all strings that form a valid signature of $C_i$ under $\Sigma$. $\ws(\sam)$ is hard because $\Sigma$ is existential-unforgeable. 


$\hyd_2$ is obtained by a small change in $\hyd_1$. Instead of sampling a random $C_i$, it is obtained by computing $h(m_i,r_i)$ from random $(m_i,r_i)$. 
This change only causes (statistically) a negligible error. This is because if $h\gets \calH$ and $r_i\gets \calR$ then $C_i:=h(m_i,r_i)$ will be uniformly random by the uniformity property of $\calH$. In addition the chameleon property of $\calH$ tells us that $r_i\gets h^{-1}_{td}(C_i,m_i)$ is distributed statistically close to sampling $r_i\gets \calR$. Therefore the order of generating $C_i$ and $r_i$ does not matter. 

Thus we see that $\calB$ is able to break the collision-resistance property of the Chameleon hash function. 

\fsnote{It seems that the starting EU scheme doesn't need to be secure against adaptive chosen message attacks. Note that we generate $c_i$ and obtain $\sigma_i$ all in advance,  so maybe it suffices to have weak security under static attack in which the adversary non-adaptively makes all its queries before seeing the public key. (Actually CHKP only proved static security of their Bonsai tree scheme and referred to a standard transformation that amplifies static security to adaptive security using Chameleon hash [KR00]. It's not surprising that, additionally with RO, TOO transformation directly gets adaptive-SU from static-EU.)}
\end{itemize}

In sum, we have shown that if there is an adversary $\calA$ breaking $\Sigma'$, then there is an adversary who manages to break either the collision resistance of the chameleon hash function $\calH$, or the existential unforgeability of the original signature scheme $\Sigma$ with probablity $\Omega(\veps)$. This contradicts the security of $\Sigma$ and $\calH$ if $\veps\geq 1/{poly(n)}$. Thus we conclude that Theorem~\ref{thm:eu2su} holds. 

\section{Discussion}
\label{sec:appdis}

\mypar{Obtaining a quantum-safe $\sucma$ signature scheme} In~\cite{CHKP12}, the authors presented a scheme for generating chameleon hash functions, based off the short integer solution problem for lattices. They also demonstrate a reduction showing an efficient algorithm to break the collision resistance of the hash function implies an efficient algorithm to break the short integer solution problem for lattices. Using results from ~\cite{Son14} this reduction can be shown to carry through to the quantum setting. As this problem is currently believed to be hard even for quantum computers, these chameleon hash functions' collision resistance remains even when faced with a quantum adversary. This chameleon hash function scheme can therefore be used in the transformation in this paper to get a quantum-secure transformation. This transformation, used with any quantum-safe $\eucma$ signature scheme will give a quantum-safe $\sucma$ scheme in the quantum random-oracle model.

When implementing the scheme with the chameleon hash function from \cite{CHKP12} we can see what the overhead would be in an actual realization. Let $n \geq 1, q \geq 2$, and $m = O(n \log q) $. Let $k$ be the output length of the hash function. Then the public key, $pk'$ will now carry with it a $\mathbb{Z}_q^{n \times m}$ matrix, so $| pk'| = |pk| + n(k+m)$. The secret key now includes a specialized lattice basis, which can be written as an $m \times m$ matrix over $\mathbb{Z}_q$, giving us $|sk'| = |sk| + m^2$. Finally, the signature overhead is the inclusion of a vector in $\mathbb{Z}_q^m$, so $|\sigma'| = |\sigma| + m$.

A signature scheme based off the Short Integer Solution problem for lattices is also presented in \cite{CHKP12}. Examining the proof presented there with tools from \cite{Son14}, we can see that this signature scheme is quantum-safe $\eucma$. Applying this transformation to this scheme, we obtain a quantum-safe $\sucma$~signature scheme. In fact, we can show that the reduction shown in \cite{CHKP12} is not as tight as it could be, and for a message of length $k$ and at most $Q$ queries, we can show that for adversary $\mathcal{F}$ and reduction $\mathcal{S}$, we have that $\adv_{SIS}(\mathcal{S}^{\mathcal{F}}) \geq \adv(\mathcal{F})^{\eucma}_{SIG} / (Q(k- \log Q))$. This is a small improvement over the result of the paper, showing that $\adv_{SIS}(\mathcal{S}^{\mathcal{F}}) \geq \adv(\mathcal{F})^{\eucma}_{SIG} / (Q(k-1)+1)$

\ifnum\final=0
\mypar{Quantum-Secure SU based on lattice problems} 
\fi

\mypar{Future directions} Our work has studied a very specific transformation that gives a systematic way of getting quantum-safe $\sucma$ signatures. There are a few more transformations in the plain model (i.e. without a random-oracle)~\cite{SPW06,LKZW08,HWZ07,HWLZ08}. We conjecture that they also hold against quantum adversaries. If this is the case, it will be meaningful to evaluate all these transformations and figure out which one is preferable under specific applications. On the other hand, we chose the Bonsai-tree signature scheme~\cite{CHKP12} to instantiate the $\too$ transformation. 
There are many recent improvements on lattice-based signatures in terms of key size and computational efficiency~\cite{Boy10,MP12,DM14}, which are shown to be $\eucma$ classically. If they can be shown to be quantum-safe, they we can get more efficient quantum-safe $\sucma$ schemes in the quantum random-oracle model.

\section*{Acknowledgements}  
The authors are grateful to Andrew Childs for helpful discussions. EE was supported by NSERC on an undergraduate research award at the Institute for Quantum Computing, University of Waterloo.  FS acknowledges support from NSERC, CryptoWorks21, ORF and US ARO.



\ifnum\final=0
\newpage
\fi


\ifnum\lip=0
\bibliographystyle{alpha}
\fi
\bibliography{sign}

\appendix

\section{Classical Proof}
\label{app:cproof}

Let $\mathcal{A}$ be the forger, $\mathcal{B}$ the reduction, and $\mathcal{C}$ be the challenger. In each case, $\mathcal{B}$ and $\mathcal{A}$ will be playing a game of strong unforgeability. Let the probability that $\mathcal{A}$ succeeds be $\epsilon$. In Case 1, $\mathcal{C}$ and $\mathcal{B}$ will play a game of existential unforgeability on the signature scheme $\sigma$. In case 2, $\mathcal{C}$ and $\mathcal{B}$ will play a game of collision resistance on the chameleon hash function $h$. We show that if the probability $\mathcal{A}$ succeeds in her forgery is $\epsilon$, then the probability that $\mathcal{B}$ succeeds is $\geq \frac{1}{2} \epsilon - \negl(n)$. At the beginning of the reduction, $\mathcal{B}$ will flip a coin, and guess which case the adversary's forgery will fall under. Clearly, $\mathcal{B}$ will be correct with probability $\frac{1}{2}$. 

In our reduction, let the forgery that $\mathcal{A}$ eventually submits be $(M^*, \sigma'^*= (\sigma^*, r^*))$ Let $C^* = h( \mathcal{O}(M^* || \sigma^* ), r^* )$. Similarly, for each $M_i$ the forger submits to the signing oracle for signing, there is an associated $\sigma'_i$ and $C_i$.

\begin{itemize}
\item \textbf{Case 1}: $C^{*} \neq C_{i}$ for all $i$. We show that whenever the forger succeeds in creating a valid forgery of this type, the reduction succeeds in breaking the existential unforgeability of the original scheme $\Sigma = ( G, S, V)$.

 $\mathcal{C}$ and $\mathcal{B}$ will be playing a game of existential unforgeability, while $\mathcal{B}$ and $\mathcal{A}$ will be playing a game of strong unforgeability. We will show that whenever $\mathcal{A}$ wins her game, $\mathcal{B}$ wins his (so long as the forgery is of the type described above).

The games will play out as follows:

Firstly, $\mathcal{B}$ will act as the random oracle for $\mathcal{A}$. In the first case at least (and this will change only slightly case to case), he can do this in the following way. Whenever $\mathcal{A}$ queries the random oracle with a query, $\mathcal{B}$ looks up in a maintained table if that query has been made before. If it has, he responds with the value he responded with before. If it has not, he generates a random number and responds with that.

Now we discuss how the game of strong unforgeability transpires.

$\mathcal{C}$ sends $\mathcal{B}$ a public key $pk$ from the $\Sigma$ scheme. $\mathcal{B}$ will generate a chameleon hash function $h$, (with corresponding trapdoor $td$) and send the public key and hash function to $\mathcal{A}$ as $pk'=(pk,h)$.

$\mathcal{A}$ will start submitting messages $M_i$ to $\mathcal{B}$ for signing. For each query, $\mathcal{B}$ does the following:

\begin{itemize}
\item Choose a random $\tilde{m}_i$ and $\tilde{r}_i$ and compute $C_i = H( \tilde{m}_i, \tilde{r}_i )$
\item Sign $C_i$ by submitting it to $\mathcal{C}$ as a signing query, obtaining $\sigma_i$
\item Query $M_i || \sigma_i$ to the random oracle, obtaining $m_i = \mathcal{O}( M_i || \sigma_i )$
\item Using the trapdoor information $td$, find an $r_i$ such that $h(m_i , r_i ) = C_i$.
\item $\sigma'_i = ( \sigma_i , r_i)$
\item Send $\sigma'_i$ to $\mathcal{A}$
\end{itemize}

Eventually, $\mathcal{A}$ will submit a valid forgery $M^*, \sigma'^* = (\sigma^*, r^*)$.

Then, $\mathcal{B}$ takes these, and computes $C^* = h( \mathcal{O}(M^* || \sigma^*), r^* )$.

Noting that $C^* \neq C_i$ for all $i$, and the $C_i$'s are precisely what was submitted to $\mathcal{C}$ for signing queries, and finally, seeing as this is a valid forgery, so $V( C^*, \sigma^*) = 'accept'$, we can see that $\mathcal{B}$ submits $C^*, \sigma*$ as a valid new forgery, breaking the existential unforgeability of $\Sigma$ and winning his game with $\mathcal{C}$.

Thus in this case whenever $\mathcal{A}$ succeeds, so does $\mathcal{B}$, and so the probability $\mathcal{B}$ succeeds given we are in this case is $\epsilon$.

\item \textbf{Case 2}: This case is defined as occurring when $C^* = C_i$ for some $i$. In this case we will show a reduction to break the collision resistance of the chameleon hash function.

To start with, $\mathcal{C}$ sends $\mathcal{B}$ the description of a chameleon hash function $h$, which $\mathcal{B}$ will find a collision for.

$\mathcal{B}$ then runs the key generation algorithm of the signature scheme $\Sigma$, obtaining $(pk,sk)$. He then sends $pk' = (pk, h)$ to $\mathcal{A}$.

For each signing query $M_i$ that $\mathcal{A}$ sends to $\mathcal{B}$, $\mathcal{B}$ does the following:

\begin{itemize}
\item Choose a random $m_i$ and $r_i$ and compute $C = h(m_i, r_i)$
\item Sign $C_i$ using the signing algorithm $S$, obtaining $\sigma = S(C, sk)$
\item Reprogram the random oracle so that $\mathcal{O}(M_i || \sigma_i ) = m_i$.
\item $\sigma'_i = (\sigma_i, r_i)$
\item Send $\sigma'_i$ to $\mathcal{A}$.
\end{itemize}

Note that we have now permitted $\mathcal{B}$ to reprogram the random oracle for the purposes of this proof. Thus it is necessary to show that $\mathcal{A}$ will still output a valid forgery.

When $\mathcal{A}$ eventually submits her forgery, $(M^*, \sigma^*)$, we can see that $C^* = C_i$ for some $i$. This implies that $h( \mathcal{O}(M_i || \sigma_i ), r_i ) = h( \mathcal{O}(M^* || \sigma^* ) , r^* )$ for that $i$. This shows us a collision for the chameleon hash function $h$, which is what $\mathcal{B}$ is looking for. But we must take care to ensure that it isn't a trivial collision.

Note that $(M_i, \sigma_i , r_i ) \neq ( M^*, \sigma^*, r^* )$, simply because both the message and signature of the forgery can't be the same as that of one of the $M_i$'s. So at least one of these values is different. 

If $r_i \neq r^*$, we are done. Otherwise, it must be the case that $M^* || \sigma^* \neq M_i || \sigma_i $. In this case, since the values for the random oracle are chosen uniformly at random, with overwhelming probability, $\mathcal{O}( M^* || \sigma^* ) \neq \mathcal{O}( M_i || \sigma_i )$, giving $\mathcal{B}$ a collision for $h$.

So in this case, $\mathcal{B}$ will succeed as long as $\mathcal{A}$ does up to a negligible probability by Lemma~\ref{lemma:cind}. So the probability $\mathcal{B}$ succeeds is $\geq \epsilon - \negl(n)$
\end{itemize}

\begin{lemma}
For a forger $\mathcal{A}$, let $\mathcal{B}_1$ and $\mathcal{B}_2$ be as below, and have them play a game of strong unforgeability with $\mathcal{A}$. Then $$| Pr_{\mathcal{B}_1}( \mathcal{A} \text{ wins}) - Pr_{\mathcal{B}_2}( \mathcal{A} \text{ wins}) | \leq \negl(n),$$ as long as the underlying signature scheme is existentially unforgeable.
\label{lemma:cind}
\end{lemma}

$\mathcal{B}_1$ is defined to operate exactly as the transformation dictates. $\mathcal{B}_2$ will operate as $\mathcal{B}$ was defined to in Case 2 above.
\begin{proof}
Say the difference in probability that $\mathcal{A}$ wins was not negligible. As the distribution of all values is the same, the only difference from $\mathcal{A}$'s perspective was that the value of $\mathcal{O}( M_i || \sigma_i )$ was changed for each $i$.

But clearly the only way to have the information that they changed is if $\mathcal{A}$ had already queried $\mathcal{O}(M_i || \sigma_i)$. But if $\mathcal{A}$ does this with non-negligible probability, then we could construct a reduction to break the existential forgeability of the signature scheme by playing strong unforgeability with $\mathcal{A}$, and before submitting each $C_i$ to the signing oracle, checking to see if $\mathcal{A}$ had queried $M_i || \sigma_i$ to the random oracle. With non-negligible probability, the reduction finds a $\sigma_i$ that is a valid forgery. So he submits this along with $C_i$ and has broken the existential unforgeability of the scheme.
\end{proof}

Therefore in both cases, as long as $\mathcal{B}$ successfully guesses which case the forgery will fall under, he manages to successfully break either the collision resistance of the chameleon hash function $h$, or the existential unforgeability of the original signature scheme $\Sigma$. Since $\mathcal{B}$ correctly guesses what case he is in half of the time, his probability of success is $\geq \frac{1}{2} \epsilon - \negl(n)$.

\section{Proof of Lemma~\ref{lemma:s2ind}}
\label{sec:pfs2ind}


\begin{proof}
Let $\calA$ be an arbitrary algorithm running in $G$ (or $G'$). Consider another algorithm $B$ that runs in an experiment $EXT$ as follows: 

\begin{ncprot}{Extraction Experiment}{$EXT$}
\begin{enumerate}
	\item Generate $(pk,w,P)\gets \sam(1^n)$. Ignore $w$. 
	\item $B$ receives $pk$ and picks $j\in_R \{1,\ldots,q\}$ at random. 
	\item $B$ simulates $\calA$ on $pk$ and (quantum) access to $f$. Just before $\calA$ making the $j$th query to $f$, $B$ measures the register that contains $\calA$'s query. Let $z$ be the measurement outcome. \\
\end{enumerate}
\end{ncprot}
Let $p_B:=\Pr_{EXT}[z\in W_{pk}]$ be the probability that the output of $E$ is a valid witness. Let $\epsilon := \left| \Pr_{G}[b=1] - \Pr_{G'}[b=1] \right|$. In both experiment $G$ and $G'$, $pk$ is selected at random according to $\sam$. Let $P_{pk}$ be the probability that $pk$ is outputted. Then

\begin{align*}
\epsilon = & \left| \Pr_{G}[b=1] - \Pr_{G'}[b=1] \right| \\
= & \left| \sum_{pk} \Pr_{G}[b=1 | pk ] \cdot P_{pk} - \sum_{pk} \Pr_{G'} [b = 1 | pk] \cdot P_{pk} \right| \\
= & \sum_{pk} P_{pk} \left| \Pr_{G} [b=1 | pk ] - \Pr_{G'} [b = 1 | pk ] \right| \, .
\end{align*}
Let $\epsilon_{pk} := \left|\Pr_{G} [b=1 | pk ] - \Pr_{G'} [b = 1 | pk ] \right|$. Let $|\phi_i \rangle$ be the superposition of $\calA^G$ on input $pk$ when the $i$'th query is made. Then let $q_y(|\phi_i \rangle )$ be the sum of squared magnitudes in $\calA$ querying the oracle on the string $y$. 

Let $S = [ q] \times W_{pk}$. Let $\delta_{pk} = \sum_{(i,y) \in S} q_y (| \phi^{pk}_i \rangle )$. We employ a theorem by Bennet et al.~\cite{BBBV97}, that states that $\| |\phi^{pk}_i \rangle - | \tilde\phi^{pk}_i  \rangle \| \leq \sqrt{q \cdot \delta_{pk} }$. (Here $|\tilde\phi^{pk}_i  \rangle$ is defined in the same way as $| \phi^{pk}_i \rangle$ but with $G'$ rather than $G$).

The same paper~\cite{BBBV97} also bounds the probability of being able to distinguish the two states, which corresponds to our probability of distinguishing the two experiments, $\epsilon_{pk}$, telling us that $$\epsilon_{pk} \leq 4 \cdot \left\| |\phi^{pk}_i \rangle - | \tilde\phi^{pk}_i \rangle \right\| \leq 4 \sqrt{q \cdot \delta_{pk}}\, .$$

Now note that $P_B^{pk}$ (that is, the probability that $EXT$ outputs a valid witness given $pk$ is chosen) can be written as 

\begin{align*}
P_B^{pk} = & \sum_{i \in [0, q]} \left( \Pr[\text{$i$ chosen}] \cdot \sum_{(j,y) \in S : j=i} q_y ( | \phi^{pk}_j \rangle ) \right)\\
= & \frac{1}{q} \sum_{i \in [0,q]} \sum_{(j,y) \in S: j = i} q_y ( | \phi^{pk}_j \rangle ) \\
= & \frac{1}{q} \sum_{(i, y) \in S} q_y ( | \phi^{pk}_i \rangle ) = \frac{1}{q} \delta_{pk}
\end{align*}

So we can see that $\epsilon_{pk} \leq 4 q \sqrt{ P_B^{pk}}$. Then $$\epsilon = \sum_{pk} P_{pk} \epsilon_{pk} \leq 4q \sum_{pk} P_{pk} \sqrt{ P_B^{pk} } \stackrel{(*)}{\leq}  4q \sqrt{ \sum_{pk} P_{pk} P_B^{pk} } =  4q \sqrt{ P_B}\, ,$$
where (*) applies Jensen's inequality. Finally,  notice that $B$ can be viewed as an adversary in the witness-search game $\ws(\sam)$. Therefore, we conclude that $p_B\leq \negl(n)$ by the hypothesis that $\ws(\sam)$ is hard and hence $\left| \Pr_{G}[b=1] - \Pr_{G'}[b=1] \right|\leq \negl(n)$. 
\end{proof}

\end{document}